\documentclass[a4paper,12pt]{article}

\usepackage{amsmath}    
\usepackage{graphicx}   
\usepackage{verbatim}   
\usepackage{color}      
\usepackage{float}
\usepackage{amsfonts}
\usepackage{algorithmic}
\usepackage{algorithm}
\usepackage{bm}
\usepackage{multirow}
\usepackage{amsthm}
\usepackage[latin1]{inputenc}
\usepackage{tikz}
\usetikzlibrary{shapes,arrows}
\usepackage[round]{natbib}
\usepackage{lscape}
\usepackage{rotating}
\usepackage{longtable}
\usepackage{pdflscape}
\usepackage{lineno}

\newtheorem{definition}{Definition}

\newtheorem{proposition}{Proposition}

\tikzstyle{decision} = [diamond, draw, fill=blue!20, 
    text width=4.5em, text badly centered, node distance=3cm, inner sep=0pt]
\tikzstyle{block} = [rectangle, draw, fill=blue!20, 
    text width=5em, text centered, rounded corners, minimum height=4em]
\tikzstyle{line} = [draw, -latex']
\tikzstyle{cloud} = [draw, circle, node distance=2.25cm,
    minimum height=3.5em]

\usepackage{authblk}

\newcommand{\beginsupplement}{%
        \setcounter{table}{0}
        \renewcommand{\thetable}{S\arabic{table}}%
        \setcounter{figure}{0}
        \renewcommand{\thefigure}{S\arabic{figure}}%
 \setcounter{algorithm}{0}
        \renewcommand{\thealgorithm}{S\arabic{algorithm}}%
 \setcounter{equation}{0}
        \renewcommand{\theequation}{S\arabic{equation}}%
 \setcounter{section}{0}
        \renewcommand{\thesection}{S\arabic{section}}%
     }

\title{Fish should not be in isolation: Calculating maximum sustainable yield using an ensemble model}

\author[1*]{Michael A. Spence}
\author[1]{Khatija Alliji}
\author[1]{Hayley J. Bannister}
\author[1]{Nicola D. Walker}
\author[1]{Angela Muench}

\affil[1]{Centre for Environment, Fisheries and Aquaculture Science, Pakefield Road, Lowestoft, Suffolk NR33 0HT, UK}

\affil[*]{michael.spence@cefas.co.uk}
\date{}
\begin{document}
\maketitle
\textbf{Running title:} Fish should not be in isolation

\newpage

\abstract{
Many jurisdictions have a legal requirement to manage fish stocks to maximum sustainable yield (MSY). Generally, MSY is calculated on a single-species basis, however in reality, the yield of one species depends, not only on its own fishing level, but that of other species.
We show that bold assumptions about the effect of interacting species on MSY are made when managing on a single-species basis, often leading to inconsistent and conflicting advice, demonstrating the requirement of a multispecies MSY (MMSY). 
Although there are several definitions of MMSY, there is no consensus. Furthermore, calculating a MMSY can be difficult as there are many models, of varying complexity, each with their own strengths and weaknesses, and the value if MMSY can be sensitive to the model used.
Here, we use an ensemble model to combine different multispecies models, exploiting their individual strengths and quantifying their uncertainties and discrepancies, to calculate a more robust MMSY. We demonstrate this by calculating a MMSY for nine species in the North Sea. We found that it would be impossible to fish at single-species MSY and that MMSY led to higher yields and revenues than current levels.
}


\section{Introduction}

The human population is growing, which has increased the demand for food production and security, which has led to an incompatibility of food production and conservation priorities. Both of which require a balance, to sustainably support an increasing population \citep{Hilborn_07}. Marine fish are a valuable source of food and income for many countries, however the global yield has levelled off and begun to decline, since the 1990s \citep{FA02009,Worm_09}. There is now an urgent need to manage fish stocks sustainably so that the balance between food production, conservation and the socio-economics are considered. This will ideally lead to an increase in food production, whilst protecting fish stocks and jobs for future generations \citep{Mensil_12}.

Fisheries managers use maximum sustainable yield (MSY), which is intended to ensure the sustainability of fish stocks whist maximising food production, without compromising the reproductive potential of the stock \citep{Hilborn_07, Mensil_12}. The legal requirement to manage fish stocks to MSY, was adopted in 1982 at the United Nations Convention on the Law of the Sea, the EU Regulation 1380/2013 and the 2002 UN world summit on sustainable development (A/CONF.199/20). The concept of MSY has been adopted by many fisheries management organisations throughout the 1900s, where mathematical and production models were used \citep{Tsikliras_Froese}. 
Recently, MSY has become a widely used reference point in the assessment of fish stocks around the world \citep{Hilborn1992,Pauly_Froese_14,Tsikliras_Froese}, here defined as:

\begin{definition}
The fishing mortality that leads to the maximum sustainable yield of the $i$th stock is
\begin{linenomath}
\begin{equation*}
F_{MSY,i}(\bm{F}_{-i})={\text{arg}\sup}_{F_i}\left(f_{1,i}(F_i,\bm{F_{-i}})\right),
\end{equation*}
\end{linenomath}
where $F_i$ is the fishing mortality of the $i$th species, $\bm{F}_{-i}$ is the fishing mortality of the other species and $f_{1,i}(F_i,\bm{F_{-i}})$ is the $i$th species' long-term annual yield (see supplementary material).
\end{definition}
MSY has historically been centred on single species MSY \citep[SS-MSY,][]{Hart_fay}, here defined as: 
\begin{definition}
The single species fishing mortality that leads to the maximum sustainable yield of the $i$th stock, the single-species MSY (SS-MSY), is
\begin{linenomath}
\begin{equation*}
F_{MSY,i}(\bm{F}_{-i})=F_{MSY,i},
\end{equation*}
\end{linenomath}
$\forall{}\bm{F}_{-i}$.
\label{ass:ssmsy}
\end{definition}

Although MSY is now a widely used concept it often applied to single stocks and when used in this way does not provide information on ecological interactions, resulting in significant ecosystem and fishing ramifications \citep{Andersen_15,saterberg_19}. The use of MSY in fisheries management has been criticised for leading to significant changes in community structure, degradation of marine ecosystems and over-exploitation of fisheries resources \citep{Larkin_77,Hilborn_07,Andersen_15}. This has rendered MSY policy guidance as incomplete in terms of ecosystem sustainability \citep{Gaichas_08}.

\begin{proposition}
SS-MSY exists if and only if
\begin{linenomath}
\begin{equation*}
\frac{\partial{}F_{MSY,i}(\bm{F}_{-i})}{\partial{}F_j}=0,
\end{equation*}
\end{linenomath}
$\forall{}j\neq{}i$.
\label{prop:ssmsy}
\end{proposition}
\begin{proof}
See supplementary material.
\end{proof}
Proposition \ref{prop:ssmsy} suggests that the fishing mortality on other species does not affect the value of $F_{MSY,i}$, which does not seem plausible in reality and therefore Definition \ref{ass:ssmsy} is rarely met. To counteract this, there has been a push to move to ecosystem-based fisheries management \citep[EBFM,][]{Pikitch_04,link_2011}.

Although there is no universally agreed definition of multispecies MSY \citep[MMSY,][]{essington_punt,norrstrom2017nash}, there are a number of proposed alternatives such as the maximum sustainable yield of the community \citep{Andersen_book}. When satisfied, the maximum sustainable yield of the community leads to over-exploitation of species with larger body sizes, leading to a decrease of predating pressure on fish stocks with smaller body sizes \citep{Andersen_book,Andersen_15,Szuwalski_16}. Fishing in this way leads to the collapse of stocks of larger species, reducing the diversity of the ecosystem and decreasing the monetary value of the fishery \citep{Andersen_book,Andersen_15}. \citet{Thorpe_19} and \citet{saterberg_19} extended this definition to include the risk of species collapse as a caveat to the maximum sustainable community yield. Alternatively the Nash equilibrium has been used to define MMSY \citep{norrstrom2017nash,thorpe17,Farcas_Rossberg}. 
The Nash equilibrium is a solution to multi-player games where no player can improve their payoff given fixed strategies played by the opponents \citep{nash}. This was applied, with the interest of maximising the yield of each fish stock, whilst taking account of the ecological impacts on other species \citep{norrstrom2017nash}.

Multispecies reference points have been predicted by production models \citep{sissenwine_shepherd_87}, however these models ignore the effects of ecological interactions between species \citep{norrstrom2017nash}.
Mechanistic multispecies models, henceforth known as simulators, are being increasingly used to support policy decisions, including fisheries and marine environmental polices \citep{hyder,Nielsen_et_al_18}, and are able to capture these interactions. They describe how multiple species interact with their environment and one another, through mechanistic processes allowing them to better predict into the future \citep{hollowed_etal}.
However, calculating reference points can be sensitive to the choice of simulator that was used to generate them \citep{Collie_16,Essington_13,fulton_03,Hart_fay}. This choice can be arbitrary as, although some simulators are better at describing some aspects of the system than others, in general no simulator is uniformly better than the others \citep{chandler}. 
Instead of choosing one simulator, it is possible to combine them using an ensemble model, allowing managers to maximise the predicting power of the simulators, whilst reducing the uncertainty and errors, improving their decision making. 
\citet{spence_ff} developed an ensemble model that treats the individual models as exchangeable and coming from a distribution. Their model exploits each simulators strengths, whilst discounting their weaknesses to give a combined solution. 

In this paper, we demonstrate how multiple simulators can be combined to find a MMSY. 
We demonstrate it by finding a Nash equilibrium for nine species in the North Sea using the ensemble model developed in \citet{spence_ff}. Although demonstrated with this definition of MMSY in the North Sea, the procedure can be used to optimise any objective function in any environment, including single-species reference points.


\section{Methods}
\label{sec:method}
We modelled nine species (see Table \ref{tb:species}) in the North Sea using historical fishing mortality from 1984 until 2017 \citep{ices_hawg,ices_ns} and fixed fishing mortality, $\bm{F}=(F_1,\ldots{},F_9)'$, from 2017 to 2050, with $F_i\in[0,2]$ for $i=1,\ldots{},9$.
Our aim was to find $\bm{F}$ values that satisfy the Nash equilibrium \citep{nash}, 
with a probability that the spawning stock biomass (SSB) falls below $B_{lim}$, the level of SSB at which recruitment becomes impaired, of 0.25 or less (see Section \ref{sec:discussion}).
\begin{table}[ht]
\caption{A summary of the species in the model. The SS-MSY values were taken from \citet{ices_hawg} and \citet{ices_ns}.}
\label{tb:species}
\begin{center}
\begin{tabular}{lllll}
\hline
$i$ & Species & Latin name & SS-FMSY & Price per tonne (\pounds{})\\
\hline
1 & Sandeel & \emph{Ammodytes marinus} & NA & 1314.59\\
2 & Norway pout & \emph{Trisopterus esmarkii} & NA & 151.96\\
3 & Herring & \emph{Clupea harengus} & 0.33 & 528.34\\
4 & Whiting & \emph{Merlangius merlangus} & 0.15 & 785.30\\
5 & Sole & \emph{Solea solea} & 0.20 & 8387.12\\
6 & Plaice & \emph{Pleuronectes platessa} & 0.21 & 1718.21\\
7 & Haddock & \emph{Melanogrammus aeglefinus} & 0.19 & 1346.99\\
8 & Cod & \emph{Gadus morhua} & 0.31 & 1745.22\\
9 & Saithe & \emph{Pollachius virens} & 0.36 & 855.33
\end{tabular}
\end{center}
\end{table}
We defined our reference point as:
\begin{definition}
$\bm{F}_{Nash}$, is when
\begin{linenomath}
\begin{equation*}
\forall{}i,F_i:f_{1,i}(F_{Nash,i},\bm{F}_{Nash,-i})\geq{}f_{1,i}(F_{i},\bm{F}_{Nash,-i})
\end{equation*}
\end{linenomath}
and $Pr(B_{i}(\bm{F}_{Nash})<B_{lim,i}))<0.25$, where $B_{i}(\bm{F})$ is the long-term SSB of the $i$th species under future fishing mortality $\bm{F}$.
\label{def:nash}
\end{definition}

We used the ensemble model yield and SSB in 2050 to be the long-term yield, $f_{1,i}(\bm{F})$, and long-term SSB, $B_{i}(\bm{F})$, for $i=1,\ldots{},9$, respectively.
As running simulators and the ensemble model is computationally expensive, we used a Gaussian process emulator \citep{kennedy_ohagan} to describe $f_{1,i}(\bm{F})$ and the 25th percentile of the long-term SSB of the $i$th species under future fishing mortality $\bm{F}$, $f_{2,i}(\bm{F})$ (i.e. $Pr(B_{i}(\bm{F}) < f_{2,i}(\bm{F}))=0.25$), which we iteratively updated after rounds of simulations, allowing us to efficiently search for $\bm{F}_{Nash}$ values.
A round consisted of running each simulator and the ensemble model for $\bm{F}$ values to find the yield and SSB.
We ran four rounds to find the $\bm{F}$ values that satisfied Definition \ref{def:nash}. 
The first round of 196 $\bm{F}$ were chosen using Sobol' sequences, a space filling algorithm \citep{sobol}.
For the subsequent rounds we proposed 100 $\bm{F}$ values that we belied may be $\bm{F}_{Nash}$ values according to the Gaussian process emulator.

We found $\bm{F}_{Nash}$ values by the following steps:
\begin{enumerate}
    \item Generate $\bm{F}^{(l)}$, for $l=1,\ldots{}196$, using Sobol' sequences.
    \item Evaluate the simulators and the ensemble model at each of the new scenarios to find the yield and the SSB.
    \label{list:nr1}
    \item Emulate the predictions of the long-term yield, $f_{1,1:9}(\bm{F})$, and the 25th percentile of the long-term SSB from the ensemble model,  $f_{2,1:9}(\bm{F})$.
    \label{list:nr3}
    \item Find 100 potential Nash equilibria, $\bm{F}^{(l)}$ (for $l=197,\ldots{},296$ in the second round, $l=297,\ldots{},396$ in the third round and $l=397,\ldots{},496$ in the fourth round).
    \label{list:nr2}
    \item Repeat step \ref{list:nr1} to \ref{list:nr2} twice.
    \item Evaluate the simulators and the ensemble model the scenarios $\bm{F}^{(l)}$, for $l=397,\ldots{},496$, to find the yield and the SSB.
    \label{list:nr6}
\end{enumerate}

After the fourth round (step \ref{list:nr6}), the final $\bm{F}_{Nash}$ values were all of $\bm{F}^{(l)}$ scenarios that satisfy Definition \ref{def:nash}, for $l=397,\ldots{},496$.
Due to the uncertainty in the ensemble model, we had several final $\bm{F}_{Nash}$ values. To try and distinguish between these we calculated the expected revenue for each of them.

The rest of the methods are as follows: the simulators, the ensemble model and the Gaussian process emulator are described in Sections \ref{sec:sim}, \ref{sec:en_mod} and \ref{sec:gpe} respectively; an algorithm of how we find potential $\bm{F}_{Nash}$ values is described in Section \ref{sec:nash} and we conclude by describing how we calculated the revenue of the long-term yield in Section \ref{sec:revenue}.


\subsection{Simulators}
\label{sec:sim}
Four multispecies simulators were used: EcoPath with EcoSim (EwE \citet{mackinson}), LeMans \citep{thorpe15}, mizer \citep{blanchard} and FishSUMs \citep{fishsums}. All of them the simulators were able to describe the dynamics of
all nine species with the exception of FishSUMs, which did not model sole. 

To keep the interpretation of fishing mortality the same across simulators, we used the single-species assessments fishing mortality at age to drive the dynamics of the simulators \citep{ices_hawg,ices_ns}.
For the size-based simulators, LeMans, mizer and FishSUMs, we calculated the length at age using their respective von Bertalanffy parameters, however for EwE we used the $\bar{F}$ 
values from the assessments. In the future (2018-2050), the age selectivity was the same as those in 2017 and species that appear in the models but not in the study were fished at their 2017 levels.

\subsection{Ensemble model}
\label{sec:en_mod}
The predicted yields and SSB from the four multispecies simulators were combined using the ensemble model of \citet{spence_ff}. 
The ensemble model is described in Table \ref{tb:ensemble_mod} and the simulator specific values are described in Table \ref{tb:simulators}.
We fit two ensemble models, one for the yields ($j=1$) and one for the SSB ($j=2$). For the yields, the simulators and the observations are in natural log tonnes, with the observations, $\hat{\bm{y}}^{(t)}_1$,
coming from \citet{ices_landings}.
In the SSB model ($j=2$), the simulators predicted SSB and the observations were in natural log tonnes, for their respective species and years, with the observations, $\hat{\bm{y}}^{(t)}_2$, coming from stock-assessments \citep{ices_hawg,ices_ns}. 

Due to the high dimensionality and correlation of the uncertain parameter space, we fitted the ensemble model using No U-turn Hamiltonian Monte Carlo \citep{hoffman_gelman} in the package Stan \citep{stan}. We ran the algorithm for 2000 iterations discarding the first 1000 as burn in.

\begin{table}[ht]
\caption{A summary of the variables in the ensemble model. The ensemble model is run for 1984--2050. For values of $n_k$, $M_k$ and $T_k$ see Table \ref{tb:simulators}.}
\label{tb:ensemble_mod}
\begin{center}
\begin{tabular}{cclp{5cm}l}
\hline
Variable&Dimensions&$t$ & Description & Relationship \\
\hline
$\bm{y}^{(t)}_j$&$9$&1984--2050&The truth & $\bm{y}^{(t)}_j\sim{}N(\bm{y}^{(t-1)}_j,{\Lambda}_{y,j})$\\
$\bm{\hat{y}}^{(t)}_j$&$9$&1984--2017&Noisy observation of $\bm{y}^{(t)}_j$&$\hat{\bm{y}}_j^{(t)}\sim{}N({\bm{y}}^{(t)}_j,\Sigma_{y,j})$\\
$\bm\delta_j$&$9$&NA&Long-term shared discrepancy&\\
$\bm\eta^{(t)}_j$&$9$&1984--2050&Short-term shared discrepancy&$\bm{\eta}_j^{(t)}\sim{}N(
    R_{\eta,j}\bm{\eta}_j^{(t-1)},\Lambda_{\eta,j}$)\\
$\bm\mu^{(t)}_j$&$9$&1984--2050&Simulator consensus&$\bm\mu^{(t)}_j=\bm{y}^{(t)}_j + \bm\delta_j + \bm\eta^{(t)}_j$\\
$\bm{\gamma}_{k,j}$&$9$&NA&Simulator $k$'s long-term individual discrepancy&$\bm\gamma_{k,j}\sim{}N(\bm{0},C_{\gamma,j})$\\
$\bm{z}_{k,j}^{(t)}$&$9$&1984--2050&Simulator $k$'s short-term individual discrepancy&$\bm{z}_{k,j}^{(t)}\sim{}N(R_{k,j}\bm{z}_{k,j}^{(t-1)},\Lambda_{k,j})$\\
$\bm{x}_{k,j}^{(t)}$&$9$&1984--2050&Simulator $k$'s best guess&$\bm{x}_{k,j}^{(t)}=\bm{\mu}_j^{(t)} + \bm{\gamma}_{k,j} +\bm{z}_{k,j}^{(t)}$\\
$\bm{\hat{x}}_{k,j}^{(t)}$&$n_k$&$T_k$& The expectation of simulator $k$'s output $\bm{x}_{k,j}^{(t)}$&$\bm{\hat{x}}_{k,j}^{(t)}\sim{}N(M_k\bm{x}_{k,j}^{(t)},\Sigma_{k,j})$
\end{tabular}
\end{center}
\end{table}

\begin{landscape}
\begin{longtable}{p{0.5cm}p{2.5cm}p{2.5cm}p{1.5cm}p{3cm}p{4cm}p{3cm}}
\caption{A summary of the simulators, their outputs used in the case study, the simulator-specific values of $n_k$, $T_k$, $M_k$ and $\Sigma_k$.}
\label{tb:simulators}\\
\hline
$k$ & Simulator & Description & $n_k$ & $T_k$ & $M_k$ & Reference for $\Sigma_k$ \\
\hline
1 & EcoPath with EcoSim (EwE) &An ecosystem model with 60 functional groups for the North Sea &
$n_1 = 9$ &
$T_1=1991-2050$ & 
\begin{equation*}
M_{1}=I_9
\end{equation*}
&
\citet{mackinson}\\
2 & LeMans&
Abundance in length classes is modelled by species&
$n_2 = 9$ &
$T_2=1986-2050$&
\begin{equation*}
M_{2}=I_9.
\end{equation*}
&
\citet{thorpe15}\\
3 & mizer &
Total weight is modelled in weight classes by species &
$n_3 = 9$ &
$T_3=1984-2050$
&
\begin{equation*}
M_{3}=I_9
\end{equation*}
&
\citet{spence_ns}\\
4 & FishSUMs&
Abundance in length classes is modelled by species&
$n_4 = 8$ &
$T_4=1984-2050$.&
\begin{equation*}
 M_4=\begin{pmatrix}
  1 & 0 & 0 & 0 & 0 & 0 & 0 & 0 & 0 \\
  0 & 1 & 0 & 0 & 0 & 0 & 0 & 0 & 0 \\
  0 & 0 & 1 & 0 & 0 & 0 & 0 & 0 & 0 \\
  0 & 0 & 0 & 1 & 0 & 0 & 0 & 0 & 0 \\
  0 & 0 & 0 & 0 & 0 & 1 & 0 & 0 & 0 \\
  0 & 0 & 0 & 0 & 0 & 0 & 1 & 0 & 0 \\
  0 & 0 & 0 & 0 & 0 & 0 & 0 & 1 & 0 \\
  0 & 0 & 0 & 0 & 0 & 0 & 0 & 0 & 1 \\
 \end{pmatrix}
\end{equation*}
&
\citet{spence_ff} 
\end{longtable}
\end{landscape}


\subsection{Gaussian process emulator}
\label{sec:gpe}

The four simulators ran $m$ future fishing scenarios, $\bm{F}^{(l)}$ for $l=1,\ldots{},m$, from 2018 until 2050.
The ensemble model was evaluated at each of these future fishing scenarios to find the long-term yield, $f_{1,i}(\bm{F})$, and long-term SSB. 
To find the Nash equilibrium we were required to evaluate $f_{1,i}(\bm{F})$ and the 25th percentile of the long-term SSB, $f_{2,i}(\bm{F})$, of the $i$th species at all $\bm{F}$ values. However, this was practically infeasible, as the simulators are relatively slow to run due to the computational complexity.



We used a Gaussian process emulator \citep{kennedy_ohagan,noe_19} to estimate $f_{1,i}(\bm{F})$ and $f_{2,i}(\bm{F})$ for all $\bm{F}$ values. If we let $\bm{f}_{j,i}=\left(f_{j,i}(\bm{F}^{(1)}),f_{j,i}(\bm{F}^{(2)}),\ldots,f_{j,i}(\bm{F}^{(m)})\right)'$, for $j=1$ and $2$, then we say that
\begin{linenomath}
\begin{equation*}
    \bm{f}_{j,i}\sim{}GP(\bm{\eta}_{j,i},K_{j,i}),
\end{equation*}
\end{linenomath}
where $\bm\eta_{j,i}$ was a generalised additive model \citep{wood_gam} and $K_{j,i}$ was the Matern covariance function (for more details see supplementary material), fitted using the DiceKriging package \citep{DiceKriging} in R \citep{R}.

\subsection{Finding the Nash equilibrium}
\label{sec:nash}

An algorithm to find the Nash equilibrium is to iteratively update $F_i$ by solving $F_i=F_{MSY,i}(\bm{F}_{-i})$ \citep{thorpe17,norrstrom2017nash}. 
At each iteration, the long-yield term yield of the $i$th species from the ensemble model was maximised by changing $F_i$.
In our case this meant 
maximising $f_{1,i}(F_i,\bm{F}_{-i})$, a stochastic function, such that $f_{2,i}(F_i,\bm{F}_{-i}) > B_{lim,i}$. 
At each iteration, we sampled 500 potential $F_i$ values, using a  Latin hypercube \citep{McKay_79}, and used the emulators to predict the long-term yield of the $i$th species and the 25th percentile of the long-term SSB of all species.
$F_i$ then became the potential new value with the largest estimated long-term yield for the $i$th species such all of the species' 25th percentile of their long-term SSB's were above their respective $B_{lim}$'s. This is summarised in Algorithm \ref{alg:pmda}.

\begin{algorithm}
\caption{A single iteration to find the Nash equilibrium. $LHC_{500}(0,2)$ is 500 samples from a Latin hypercube of 1 dimension.}
\label{alg:pmda}
\begin{algorithmic}
\FOR {$i$ in $1:9$}
\STATE $\bm{F}'\sim{}LHC_{500}(0,2)$
\STATE $\bm{{\tilde{f}}}_{i,1}\sim{}GP(\bm{\eta}_{1,i}(\bm{F}',\bm{F}_{-i}),K_{1,i})$
\STATE $\bm{{\tilde{f}}}_{1:9,2}\sim{}GP(\bm{\eta}_{2,i}(\bm{F}',\bm{F}_{-i}),K_{2,i})$
\STATE $ll\gets{}{\text{arg}\max}_l\left\{{{\tilde{f}}_{i,1,l}}:\tilde{f}_{i',2,l} > B_{lim,i'}\text{ for }i'=1,\ldots{},9\right\}$
\STATE $F_i\gets{}{F}_{ll}'$
\ENDFOR
\end{algorithmic}
\end{algorithm}

To initialise the algorithm, we sampled 10,000 $\bm{F}$ values, using a Latin hypercube 
design and used the emulators to estimate the long-term yield of each species and the 25th percentile of the long-term SSB. The initial $F_i$ value was the proposed fishing mortality for the $i$th species that lead to the highest long-term yield such that the 0.25 percentile of the $i$th species' long-term SSB was above $B_{lim,i}$.

The Nash equilibria was estimated with 100 samples from the posterior distribution of the ensemble model by repeating Algorithm \ref{alg:pmda} between 26 and 100 times, drawn at random. The resulting 100 samples were $\bm{F}_{Nash}$ values, which we ran the simulators and the ensemble model with. 

\subsection{Revenue of the long-term yield}
\label{sec:revenue}
For the $\bm{F}_{Nash}$ values, we calculated the expected revenue from the long-term yields for each Nash equilibrium found using Algorithm \ref{alg:pmda}. To derive the revenue, we predicted the prices for each year until 2050 using a uni-variate Vector Auto-Regressive estimation model (VAR) and landings values per tonne per species (deflated) of the UK fleet in England from 1970-2018. The value of the landings per tonne are shown in Table \ref{tb:species}.


\section{Results}
\label{sec:results}

\subsection{Simulator runs}
Each simulator was run for 496 different fishing scenarios, $\bm{F}^{(l)}$ for $l=1,\ldots{}496$. 
Figure \ref{fig:histYield} shows the historical yields and each of the simulators predicted yields for the period 1985 to 2017. Most of the simulators were able to qualitatively recreate the trends of the observed yields for most of the species, however no single simulator appears to be overall better than the others.

\begin{figure}[ht]
\begin{center}
    \includegraphics{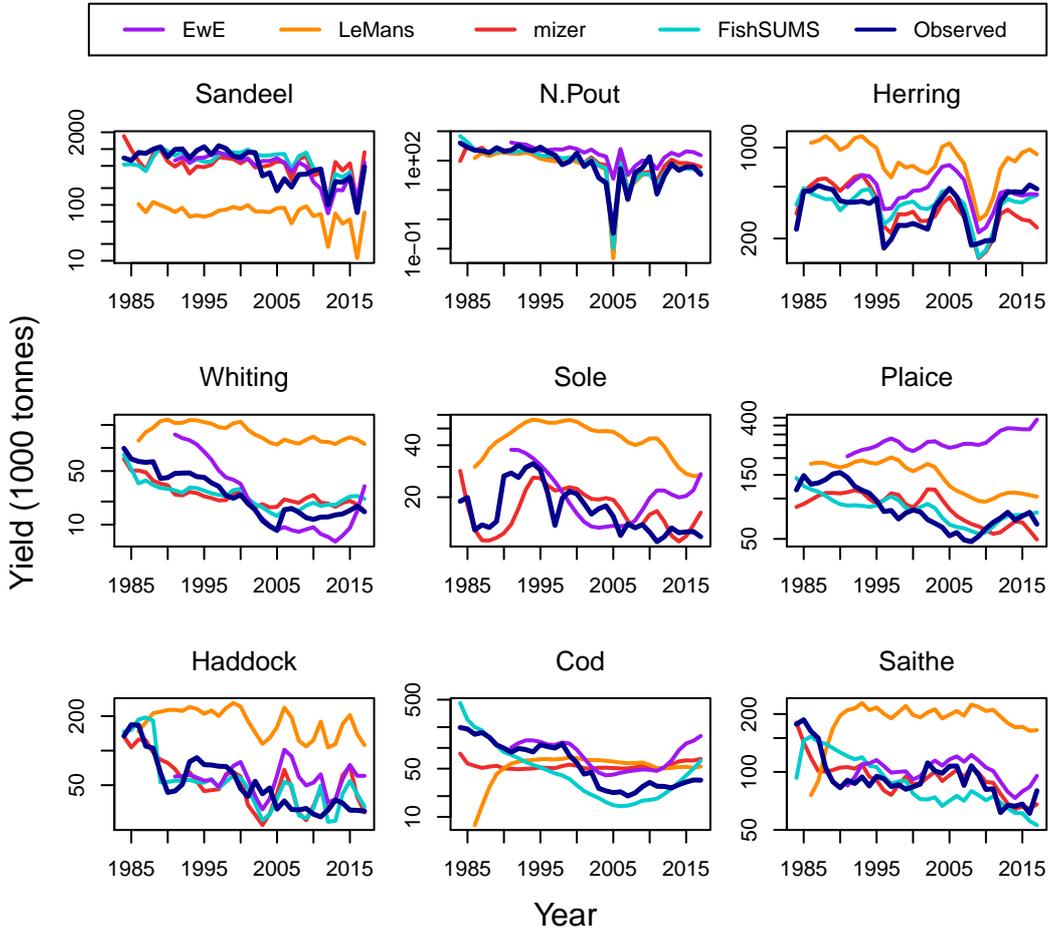}
\end{center}
\caption{Historical yield from observations \citep{ices_landings} and the simulators.}
\label{fig:histYield}
\end{figure}

\subsection{Ensemble outputs}
We fitted the ensemble model and used it to describe, with uncertainty, what the yield and SSB would be under the future fishing scenarios. The median long-term yield for all of the scenarios is shown in Figure \ref{fig:ens_Yield}. %
Although the long-term yield and SSB were sensitive to the fishing mortality of that species, it was also sensitive to the fishing mortality of other species. Figure \ref{fig:codWhiting} shows the 5th and 25th percentile of the long-term SSB for cod and whiting for varying fishing mortality of cod respectively, with the solid lines being their $B_{lim}$ values. The long-term SSB's of whiting and cod appear to be negatively correlated.

\begin{figure}[!]
\begin{center}
    \includegraphics{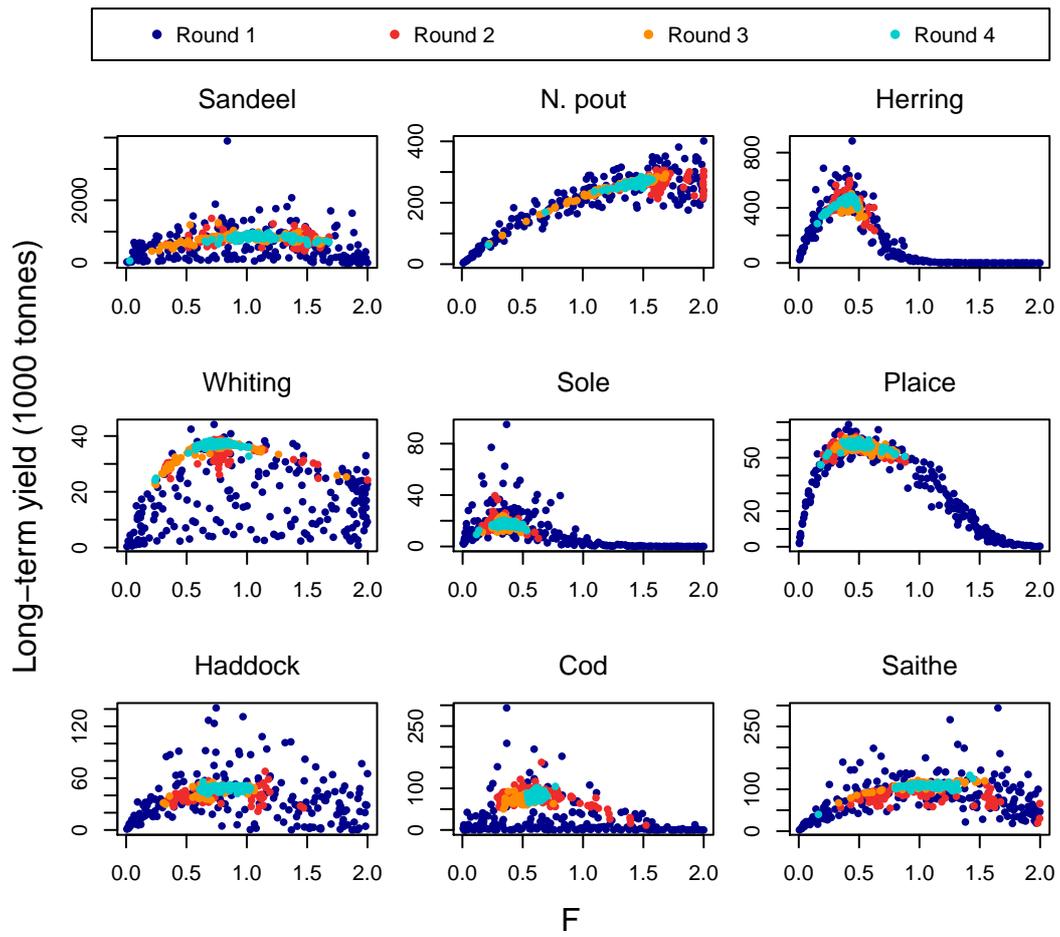}
\end{center}
\caption{The median long-term yield predicted from the ensemble model.}
\label{fig:ens_Yield}
\end{figure}

\begin{figure}[!]
\begin{center}
    \includegraphics{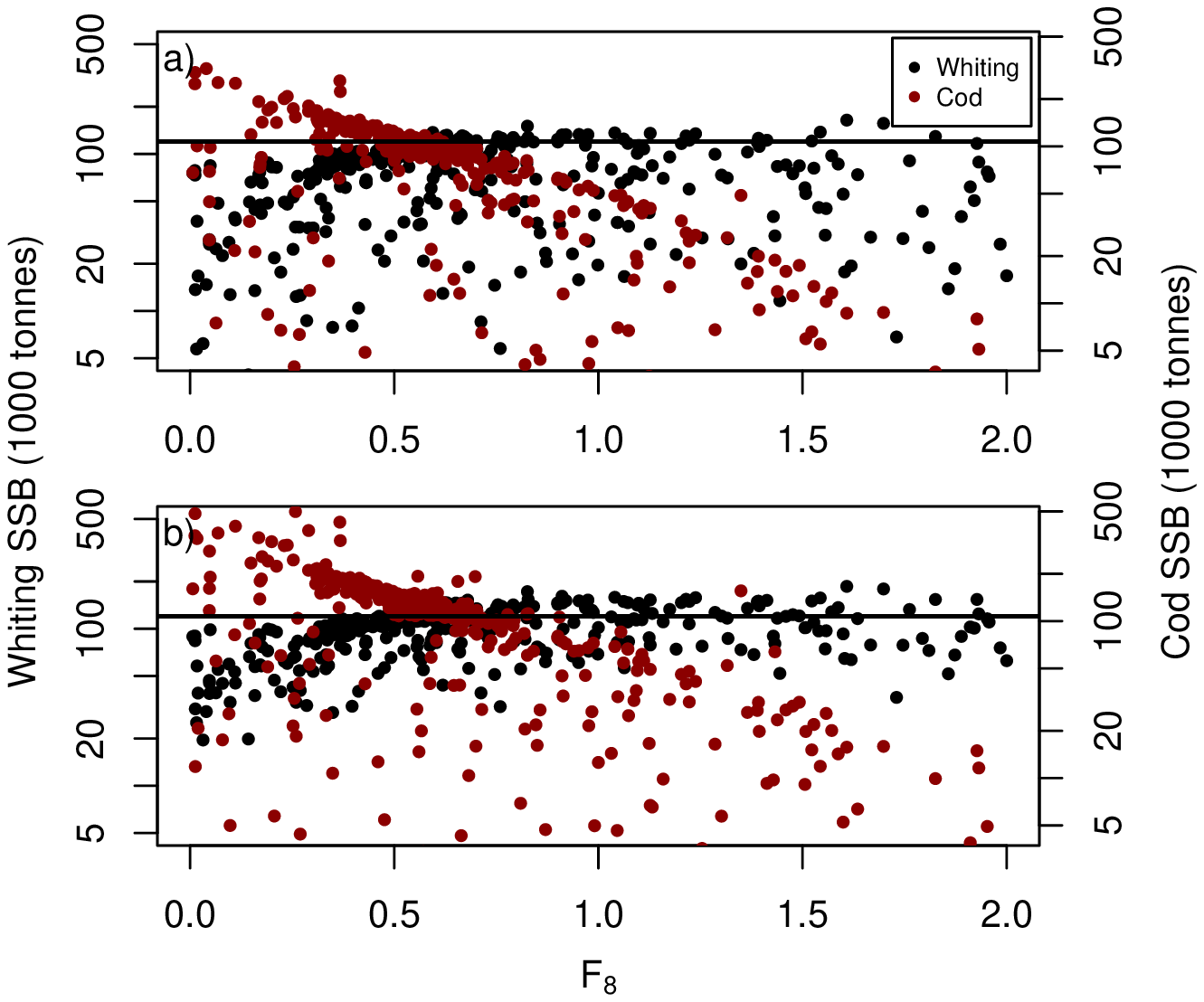}
\end{center}
\caption{The 5th (a) and 25th (b) percentile of the long-term SSB for cod and whiting under different fishing mortality rates of cod. The solid line is the $B_{lim}$ for each species.}
\label{fig:codWhiting}
\end{figure}


\subsection{Gaussian process emulator}

We fitted both the long-term yield and the 25th percentile of the long-term SSB for 100 iterations of the ensemble model. Table \ref{tb:acpt_runs} shows the number of scenarios in each round and the number of species with acceptable risk, a probability that the long-runs SSB  is above $B_{lim}$ of 0.75 or more. In later rounds the number of species with acceptable risk increases.

\subsection{Nash equilibria}

Out of the 100 potential $\bm{F}_{Nash}$ values found in the fourth round, 39 of them satisfied Definition \ref{def:nash}. For these 39 $\bm{F}_{Nash}$ values, we calculated the revenue of the long-term yield. Figure \ref{fig:final} shows the marginal distributions of the accepted $\bm{F}_{Nash}$ values, with the solid line showing the $\bm{F}_{Nash}$ value that led to the highest revenue. 
The revenues of all the $\bm{F}_{Nash}$ values were between \pounds{}1.7 billion and \pounds{}2.2 billion, larger than the revenue in 2017, \pounds{}1.3 billion. See Table S1 in the supplementary material for the 39 $\bm{F}_{Nash} values$.

\begin{figure}[!]
\begin{center}
    \includegraphics{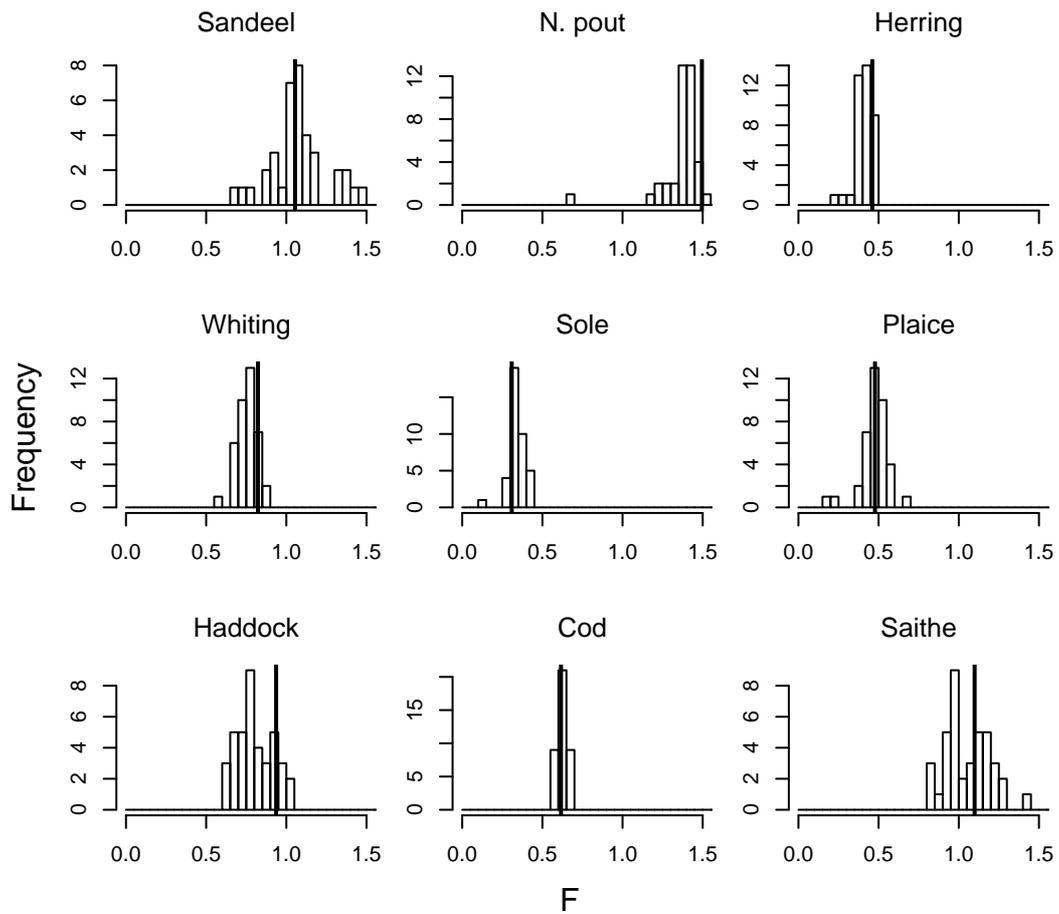}
\end{center}
\caption{The 39 Nash equilibria found in the fourth round. The solid line is the Nash equilibrium that generates the highest revenue.}
\label{fig:final}
\end{figure}

\begin{table}[ht]
\caption{The number of scenarios that have acceptable risk to the species long-term SSB. Acceptable risk to a species is that the long-term SSB is above $B_{lim}$ with a probability of 0.75 or more.} 
\label{tb:acpt_runs}
\begin{center}
\begin{tabular}{ccccc} 
\hline 
\# species & Rnd 1 & Rnd 2 & Rnd 3 & Rnd 4 \\ 
\hline
0 & $0$ & $0$ & $0$ & $0$ \\ 
1 & $14$ & $0$ & $0$ & $0$ \\ 
2 & $45$ & $0$ & $0$ & $0$ \\ 
3 & $42$ & $0$ & $0$ & $0$ \\ 
4 & $37$ & $1$ & $0$ & $0$ \\ 
5 & $31$ & $12$ & $1$ & $0$ \\ 
6 & $19$ & $33$ & $2$ & $0$ \\ 
7 & $8$ & $43$ & $28$ & $13$ \\ 
8 & $0$ & $11$ & $46$ & $48$ \\ 
9 & $0$ & $0$ & $23$ & $39$ \\ 
\end{tabular} 
\end{center}
\end{table}

\section{Discussion}
\label{sec:discussion}

In this paper, we showed that using a SS-MSY is only possible under very strict assumptions. We demonstrated how to calculate reference points using a specific definition of MMSY, the Nash equilibrium \citep{Farcas_Rossberg,norrstrom2017nash,thorpe17}, with a caveat for the risk of species collapse for nine species in the North Sea. We did this by combining multiple simulators using an ensemble model, removing the arbitrary choices of which simulator to use to calculate the reference points. We found that the Nash equilibrium led to higher fishing mortality rates than SS-MSY, leading to an increase in the long-term yields and revenue. 
To our knowledge, ensemble modelling has never been used to calculate multispecies reference points before. 

We found that the $\bm{F}_{Nash}$ values were generally higher than SS-MSY. Fishing predators at higher levels can relieve stress on prey, leading to an increase in prey, which can be exploited by the fishery \citep{Andersen_15}. These interactions are not accounted for when calculating SS-MSY, therefore adopting a MMSY means it is possible to have a higher yield \citep{Beddington_cooke, norrstrom2017nash}. 
Furthermore the revenue generated from MMSY was greater than current levels, allowing an economic gain for fishers and their families, something that is not always the case for SS-MSY \citep{Giron_2018}.

A common criteria when defining SS-MSY is that the SSB is larger than $B_{lim}$ with a probability greater than 0.95 \citep{ices_advice}. We demonstrated that the SSB of a single species not only depends on its own fishing mortality, but also the fishing mortality of the other species.
For example, only a small range of cod fishing mortality would lead to both whiting and cod's long-term SSB being above $B_{lim}$ with a probability greater than 0.75. However, no combination of $\bm{F}$ values would result in both whiting and cod's long-term SSB being above $B_{lim}$ with a probability greater than 0.95 (Figure \ref{fig:codWhiting}). This effect between whiting and cod was also found by EwE \citep{Mackinson_09} and the stochastic multispecies model \citep{sms,Kempf_10}. 
As it is impossible to find $\bm{F}$ values that satisfy the 0.95 probability criteria for all nine species in this study, we reduced our criteria to 0.75.
In general uncertainty is subjective, specific to the decision maker, the study, the information and the simulators \citep{gelman}.
In this study our certainty is limited to the simulators used, and could be reduced if they were improved, however, we were able quantify this uncertainty in a robust and interpretable manner \citep{harwood_stokes}. Currently when calculating SS-MSY, large amounts of uncertainty are ignored, e.g. species interactions, and thus estimations of probability are not robust, which makes the 0.95 caveat rather arbitrary.

%
Generally, fisheries managers select a single simulator for a species, from a set of competing simulators to calculate reference points \citep[e.g][]{ices_advice},
however, not including species interactions can lead to inconsistencies in the reference points. 
Should a manager chose a simulator for defining an MMSY, they would have to decide which simulator based on an arbitrary choice, and the values of the reference points are sensitive to the simulator (see supplementary material Figures S1-S4) \citep{Gaichas_08}.
Furthermore, choosing a simulator, without accounting for its model discrepancy \citep{kennedy_ohagan}, can lead to biased advice. For example, if we chose LeMans, sandeel yields would be consistently under-estimated, however correcting for the discrepancy would lead to estimations that were closer to the true yields (Figure \ref{fig:histYield}).
In general, no simulator is uniformly better than the others \citep{chandler}. 
In our example, mizer captures the dynamics of the saithe yields, 
however it does not capture the cod yields as well (Figure \ref{fig:histYield}).
We combined four different simulators, accounting for their discrepancies and uncertainties, to define the reference points, suggesting their values are no longer sensitive to the simulator selection \citep{spence_ff}. 

Due to the robust quantification of uncertainty in the ensemble model, we found 39 different $\bm{F}_{Nash}$ values. In practice, a manager would have to decide which of the $\bm{F}_{Nash}$ values is the `best', which is dependent on their needs and priorities. For example, they may want to maximise the total revenue or to minimise the risk to a specific species. In general, the manager's utility can be computed for the different MMSYs and then they can decide which of them is the `best'. We calculated the revenue for each $\bm{F}_{Nash}$ value, and, if we were to give advice, we would select the $\bm{F}_{Nash}$ value that would lead to the highest revenue, as shown in Figure \ref{fig:final}.

The results should be interpreted in the light of the limitations of the four simulators used in this study, which were the only ones available. Using as many simulators as possible would improve the robustness of the results, however it would be more beneficial to use better or improved simulators.
By robustly quantifying the uncertainty, the ensemble model uses all of the information from the simulators. If there was no, or very little, information in all the simulators then the ensemble model would give very uncertain predictions. 
Currently the ensemble model of \citet{spence_ff} assumes that the discrepancies of the simulators are the same in the future as they are in the past, for example a simulator that was uncertain at predicting the past would also be uncertain when predicting the future. More work is required to find the predictive power of these simulators, so we can include this information in the ensemble model.


When calculating the Nash equilibrium, we would like to use a sequential algorithm, such as in \citet{norrstrom2017nash}, which would require running the simulators and the ensemble model many times. Currently this is not feasible due to computational and time constraints caused by the simulators.
To limit the number of simulator runs required, we used a Gaussian process emulator to predict, with uncertainty, what the ensemble model would say for all future scenarios. Gaussian process emulators have been used in other fields when simulators are expensive to run \citep[e.g][]{veron,KENNEDY20061301}.
This allowed us to limit simulator runs to fishing scenarios that may be close to the Nash equilibrium and result in acceptable risk (Table \ref{tb:acpt_runs}), or where the emulator was unsure of the outcome. 
Although replacing the ensemble model with a Gaussian process leads to uncertainty in the final $\bm{F}_{Nash}$ values, this would not matter in practice, as the uncertainty Gaussian process will be small.
Using the ensemble model and the Gaussian process emulator allows for the calculation of MMSY in a robust and timely manner.




The Nash equilibrium was calculated for nine species in the North Sea, with caveats for the risk of stock collapse, although the methods described would be applicable for any definition of MMSY, or even SS-MSY, at any location. 
Alternative objectives could be ecosystem based yield \citep{Steel_11} or to aim to either maximise profits in the fisheries (e.g. using an MEY approach \citep{Dichmont16,Pascoe_18,Guillen_13}) or focus on the efficiency of the fishing practice. The latter, for example, could aim to define the reference points based on marginal value of yields, apply pareto-efficiency criteria or include joint-technology in production (i.e. mixed fisheries considerations). 
In this paper, the Nash equilibrium was chosen as it is a way of combining the SS-MSY with the MMSY as aligned concepts of MSY and EBFM \citep{norrstrom2017nash}. 


\section{Conclusion}
In this study we calculated MMSY in the North Sea using an ensemble model, demonstrating that it can lead to sustainable yields whilst ensuring ecosystem health is not diminished.
This approach can be adopted by fisheries scientists and mangers worldwide, taking account of structural uncertainties and removing arbitrary modelling decisions, leading to more robust, and therefore better science and management.
The reference points can be applied to problems in other fields such as climate science, epidemiology or systems biology.
Using the methods described in this paper, we were able to provide a practical tool to optimise any objective function for use by scientists and managers alike.

\section*{Acknowledgements}
The work was funded by the Department for Environment, Food and Rural Affairs (Defra). We would like to thank Robert Thorpe, Michaela Schratzberger and Paul Dolder for comments on earlier versions of the manuscript.

\section*{Authors contribution} MAS, HJB and KA conceived the ideas and designed the methodology; NDW, AM and MAS extracted data for the study; MAS, KA and HJB ran simulators; MAS and KA led the writing of the manuscript. All authors contributed critically to the drafts and gave final approval for publication.

\section*{Data availability statement}
Data sharing is not applicable to this article as no new data were created; rather, data were acquired from existing published sources (all sources are cited in the text), or are described, figured and tabulated within the manuscript or supplementary information of this article.

\newpage
\beginsupplement
\section{MSY}
\subsection{Function definition}
Let $y_i(F_1,F_2,\ldots{}F_n)$ be a continuous function such that
\begin{equation*}
        f_{1,i}:D\xrightarrow{}\mathbb{R}_{\geq0}
\end{equation*}
with
\begin{equation*}
D=\left\{(x_1\times{}x_2\times\ldots\times{}x_n) \in \mathbb{R}^n_{\geq0}\right\}.
\end{equation*}

\subsection{Proof of Proposition 1}
\begin{proof}
As $f_{1,i}(F_i,\bm{F}_{-i})$ is a continuous function, then $F_{MSY,i}(\bm{F}_{-i})={\text{arg}\sup}_{F_i}\left(f_{1,i}(F_i,\bm{F}_{-i})\right)$ is also a continuous function due to the maximum theorem \citep{ok_07}. 
Suppose 
\begin{equation*}
    \frac{\partial{}F_{MSY,i}(\bm{F}_{-i})}{\partial{}F_j}=0
\end{equation*}
then
\begin{eqnarray*}
    F_{MSY,i} &=& F_{MSY,i}(\bm{F}_{-i,j},F_j)\\
     &=& \lim_{\delta \to 0}F_{MSY,i}(\bm{F}_{-i,j},F_j + \delta)\\
     &=& F_{MSY,i}.
\end{eqnarray*}
Now suppose 
\begin{equation*}
    \frac{\partial{}F_{MSY,i}(\bm{F}_{-i})}{\partial{}F_j}>0
\end{equation*}
then
\begin{eqnarray*}
F_{MSY,i} &=& F_{MSY,i}(\bm{F}_{-i,j},F_j)\\
&<&\lim_{\delta \to 0}F_{MSY,i}(\bm{F}_{-i,j},F_j + \delta)\\
&=&F_{MSY,i}'
\end{eqnarray*}
hence $F_{MSY,i}\neq{}F_{MSY,i}'$. Alternatively suppose
\begin{equation*}
    \frac{\partial{}F_{MSY,i}(\bm{F}_{-i})}{\partial{}F_j}<0
\end{equation*}
then
\begin{eqnarray*}
F_{MSY,i} &=& F_{MSY,i}(\bm{F}_{-i,j},F_j)\\
&>&\lim_{\delta \to 0}F_{MSY,i}(\bm{F}_{-i,j},F_j + \delta)\\
&=&F_{MSY,i}'
\end{eqnarray*}
hence $F_{MSY,i}\neq{}F_{MSY,i}'$. Hence
\begin{equation*}
    \frac{\partial{}F_{MSY,i}(\bm{F}_{-i})}{\partial{}F_j}=0
\end{equation*}
$\forall{}j\neq{i}$ if Definition 2 is to exist.
\end{proof}

\section{Gaussian process emulator}
A stochastic process $f_{j,i}(\bm{F})$ is said to be a Gaussian process if the random vector, $\bm{f}_{j,i}=\left(f_{j,i}(\bm{F}^{(1)}),f_{j,i}(\bm{F}^{(2)}),\ldots,f_{j,i}(\bm{F}^{(n)})\right)'$, for $j=1$ and $2$ and $i=1,\ldots{}9$, has the distribution
\begin{equation*}
    \bm{f}_{j,i}\sim{}N(\bm{\eta}_{j,i},K_{j,i}).
\end{equation*}
Similarly to a multivariate Gaussian, completely specified by a mean vector and a covariance matrix, the Gaussian Process is parametried by a mean and a covariance function
with 
\begin{equation*}
    \eta_{j,i}(\bm{F})=E(f_{j,i}(\bm{F}))
\end{equation*}
and
\begin{equation*}
    k_{j,i}(\bm{F}^{(l)},\bm{F}^{(l')})=Cov\left(f_{j,i}(\bm{F}^{(l)}),f_{j,i}(\bm{F}^{(l')})\right)
\end{equation*}
respectively, returning the mean of a random variable and the covariance between two random variables, as function of the inputs only \citep{noe_19}. In this work we consider $\eta_{j,i}$ to be a generalised additive model \citep{wood_gam}, see Section \ref{sec:sup_gam}. We used the covariance $k_{j,i}(\bm{F}^{(l)},\bm{F}^{(l')})=C_{j,i,1}(F_1^{(l)},F^{(l')}_1)\otimes{}\ldots\otimes{}C_{j,i,9}(F_9^{(l)},F_9^{(l')})$ with a Mat\`{e}rn covariance function,
\begin{eqnarray*}
	C_{j,i,d}(F_d^{(l)},F_d^{(l')}) =&& \sigma^2\left(1+\sqrt{5}\frac{|F_d^{(l')}-F_d^{(l)}|}{\rho_{j,i,d}}+\frac{5}{3}\left(\frac{|F_d^{(l')}-F_d^{(l)}|}{\rho_{j,i,d}}\right)^2\right)\\
&&\times\exp\left(\frac{-\sqrt{5}|F_d^{(l')}-F_d^{(l)}|}{\rho_{j,i,d}}\right),
\end{eqnarray*}
for $d=1\ldots{}9$. 

Denote the observed data $\mathcal{D}= \left\{(\bm{F}^{(1)},y_{j,i}^{(1)}),\ldots{},(\bm{F}^{(m)},y_{j,i}^{(m)}) \right\}$ to be training data, with inputs $\bm{F}^{(l)}$ and outputs $y_{j,i}^{(l)}$ for $l=1,\ldots{},m$. The outputs are denoted $\bm{y}_{j,i}=(y_{j,i}^{(1)},\ldots,y_{j,i}^{(m)})'$. Conditioning the Gaussian process on the observed data
\begin{equation*}
    f_{j,i}(\bm{F})\sim{}GP(\tilde{f}_{j,i}(\bm{F}),s(\bm{F},\bm{F}'))
\end{equation*}
with
\begin{equation*}
    \tilde{f}_{j,i}(\bm{F})=\eta_{j,i}(\bm{F}) + \bm{k}(\bm{F})'(K + \sigma^2I)^{-1}(\bm{y}_{j,i}-\bm{\eta}_{j,i})
\end{equation*}
and
\begin{equation*}
    s(\bm{F},\bm{F}')=k(\bm{F},\bm{F}') - \bm{k}(\bm{F})'(K + \sigma^2I)^{-1}\bm{k}(\bm{F}'),
\end{equation*}
where $\bm{k}(\bm{F})=(k(\bm{F},\bm{F}^{(1)}),\ldots{}k(\bm{F},\bm{F}^{(m)}))'$, $K=\left[k(\bm{F}^{(l)},\bm{F}^{(l')})\right]^{m}_{l,l'=1}$ is the training covariance,
$\bm{\eta}_{j,i}=(\eta_{j,i}(\bm{F}^{(1)}),\ldots,\eta_{j,i}(\bm{F}^{(m)}))'$ and $I$ is the identity matrix of dimensions $m$ \citep{noe_19}.

\subsection{Generalised additive models}
\label{sec:sup_gam}
The mean function from the Gaussian process emulator was a cubic spline such that
\begin{equation*}
    s(x)=\sum_{h=1}^{H}1_{x\geq\lambda_{h}}\beta_h(x-\lambda_k)^3,
\end{equation*}
where $H$ is the number of `knots' and $\lambda_k$ is the location of the $k$th `knot'.

\subsection*{Sandeel}
The yield for sandeel was
\begin{equation*}
    \eta_{1,1}(\bm{F})=\beta_{1,1} + s(F_1) + s(F_3) + s(F_5),
\end{equation*}
and the SSB was
\begin{equation*}
    \eta_{2,1}(\bm{F})= \beta_{2,1} + s(F_1) + s(F_2) + s(F_3) + s(F_4) + s(F_5) + s(F_8) + s(F_9).
\end{equation*}

\subsection*{Norway pout}
The yield for Norway pout was
\begin{equation*}
    \eta_{1,2}(\bm{F})=\beta_{1,2} + s(F_2) + s(F_3) + s(F_5),
\end{equation*}
and the SSB was
\begin{equation*}
    \eta_{2,2}(\bm{F})= \beta_{2,2} + s(F_1) + s(F_2) + s(F_3) + s(F_8) + s(F_9).
\end{equation*}

\subsection*{Herring}
The yield for herring was
\begin{equation*}
    \eta_{1,3}(\bm{F})= \beta_{1,3} +s(F_3) + s(F_5) + s(F_8) + s(F_9),
\end{equation*}
and the SSB was
\begin{equation*}
    \eta_{2,3}(\bm{F})=\beta_{1,3} + s(F_1) + s(F_2) + s(F_3) + s(F_4) + s(F_5) + s(F_6) + s(F_8) + s(F_9).
\end{equation*}

\subsection*{Whiting}
The yield for whiting was
\begin{equation*}
    \eta_{1,4}(\bm{F})= \beta_{1,4} + s(F_3) + s(F_4),
\end{equation*}
and the SSB was
\begin{equation*}
    \eta_{2,4}(\bm{F})=\beta_{2,4} + s(F_1) + s(F_2) + s(F_3) + s(F_4) + s(F_5) + s(F_7) + s(F_8) + s(F_9).
\end{equation*}

\subsection*{Sole}

The yield for sole was
\begin{equation*}
    \eta_{1,5}(\bm{F})=\beta_{1,5} +s(F_3) + s(F_4) + s(F_5),
\end{equation*}
and the SSB was
\begin{equation*}
    \eta_{2,5}(\bm{F})=\beta_{2,5} + s(F_1) + s(F_2) + s(F_3) + s(F_4) + s(F_5) + s(F_6) + s(F_7) + s(F_8) + s(F_9).
\end{equation*}

\subsection*{Plaice}

The yield for plaice was
\begin{equation*}
    \eta_{1,6}(\bm{F})=\beta_{1,6} +s(F_1) + s(F_3) + s(F_6) + s(F_7),
\end{equation*}
and the SSB was
\begin{equation*}
    \eta_{2,6}(\bm{F})=\beta_{2,6} +s(F_1) + s(F_2) + s(F_3) + s(F_4).
\end{equation*}

\subsection*{Haddock}
The yield for haddock was
\begin{equation*}
    \eta_{1,7}(\bm{F})=\beta_{1,7} +s(F_3) + s(F_4) + s(F_7) + s(F_8),
\end{equation*}
and the SSB was
\begin{equation*}
    \eta_{2,7}(\bm{F})= \beta_{2,7} + s(F_1) + s(F_2) + s(F_3) + s(F_4) + s(F_5) + s(F_6) + s(F_7) + s(F_8).
\end{equation*}

\subsection*{Cod}
The yield for cod was
\begin{equation*}
    \eta_{1,8}(\bm{F})=\beta_{1,8} +s(F_3) + s(F_5) + s(F_8),
\end{equation*}
and the SSB was
\begin{equation*}
    \eta_{2,8}(\bm{F})= \beta_{2,8} + s(F_2) + s(F_3) + s(F_4) + s(F_5) + s(F_8).
\end{equation*}

\subsection*{Saithe}

The yield for saithe was
\begin{equation*}
    \eta_{1,9}(\bm{F})=\beta_{1,9} + s(F_3) + s(F_5) + s(F_8) + s(F_9),
\end{equation*}
and the SSB was
\begin{equation*}
    \eta_{2,9}(\bm{F})=\beta_{2,9} + s(F_1) + s(F_2) + s(F_3) + s(F_4) + s(F_8) + s(F_9).
\end{equation*}

\section{Results}
\subsection{Simulator runs}
Figures \ref{fig:ewe_res}-\ref{fig:fs_res} show the long-term yields from the simulators.

\begin{figure}[ht]
\begin{center}
\includegraphics{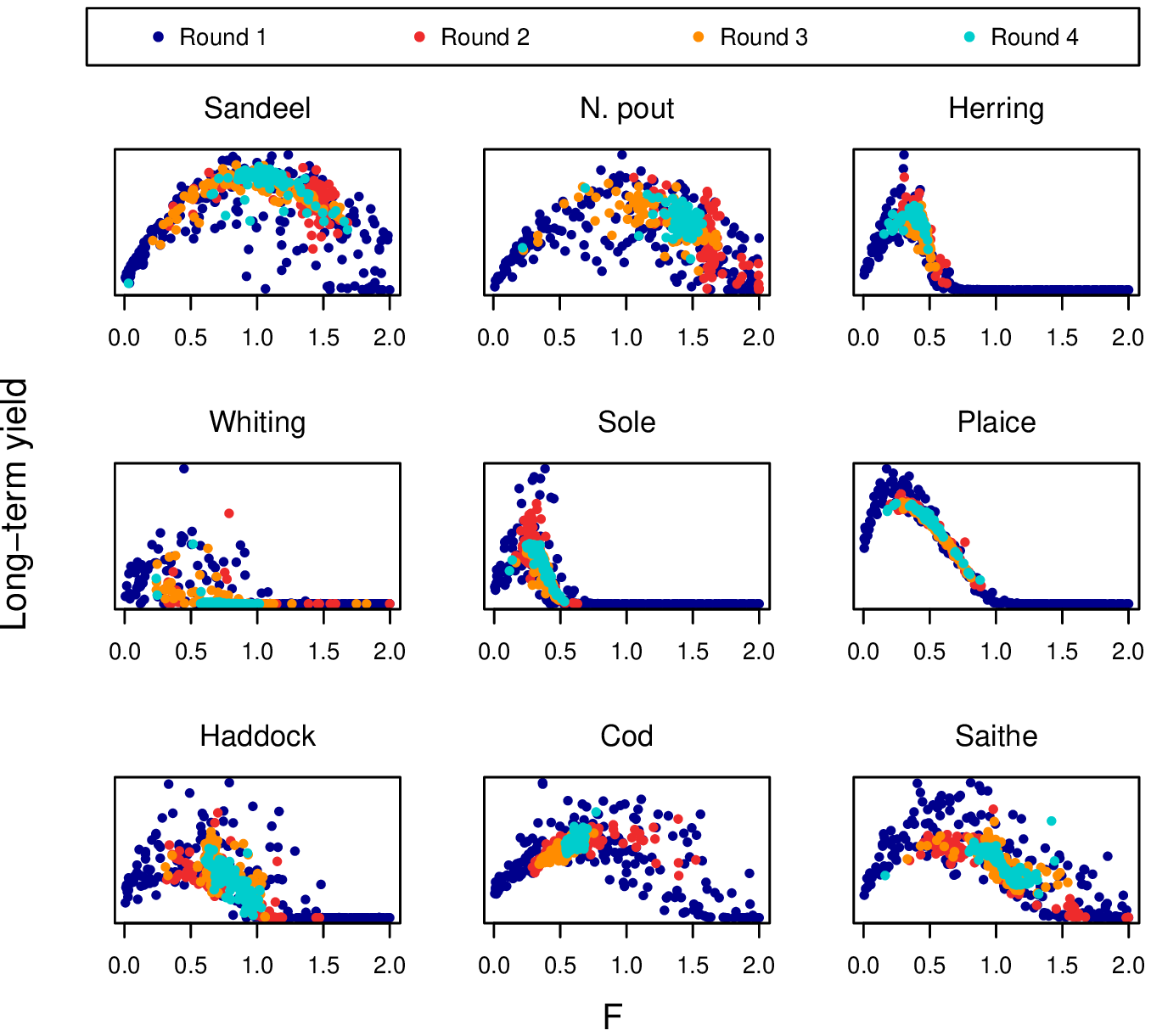}
\end{center}
\caption{The long-term yield predictions from EcoPath with EcoSim.}
\label{fig:ewe_res}
\end{figure}

\begin{figure}[ht]
\begin{center}
\includegraphics{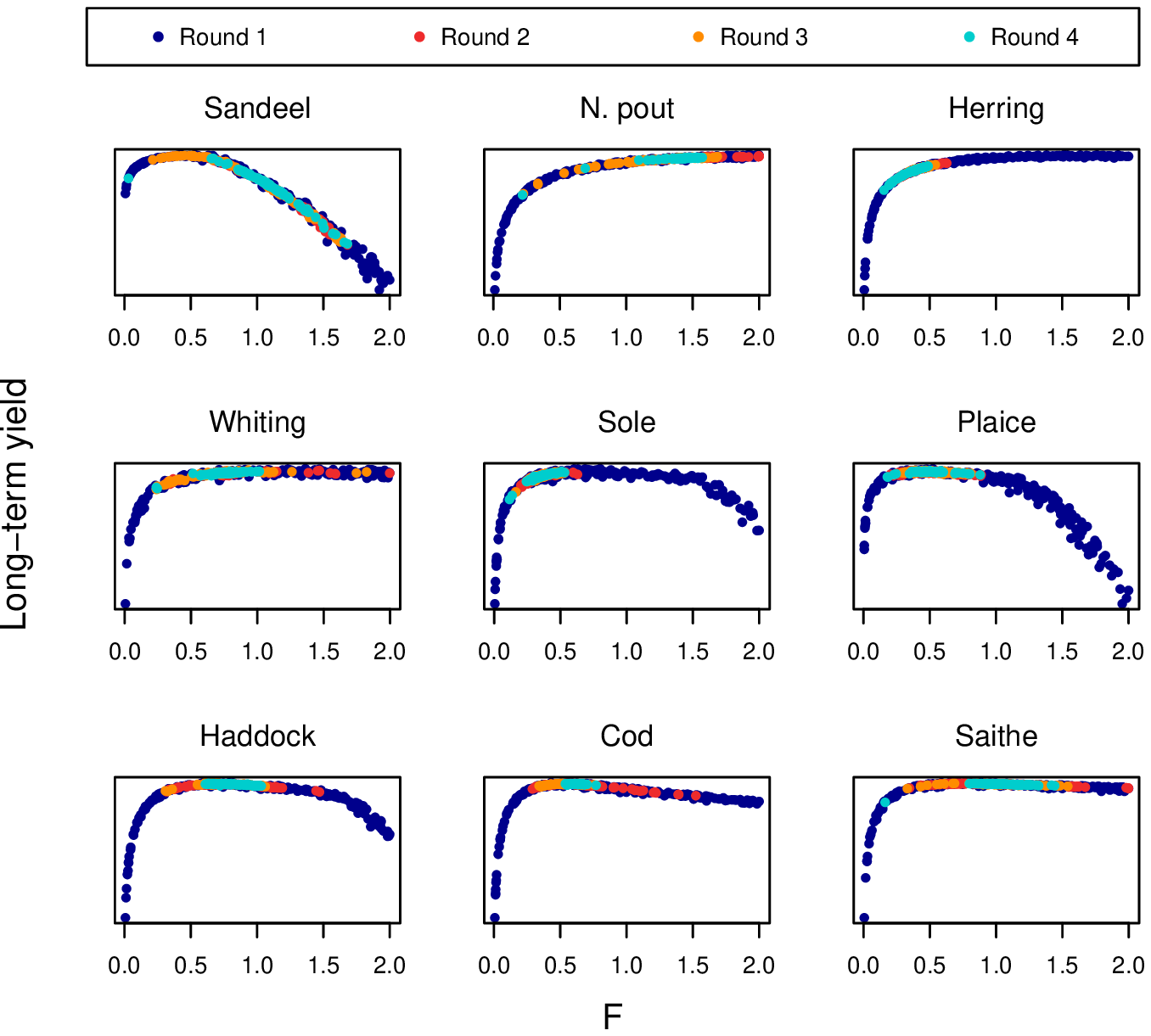}
\end{center}
\caption{The long-term yield predictions from LeMans.}
\label{fig:lm_res}
\end{figure}

\begin{figure}[ht]
\begin{center}
\includegraphics{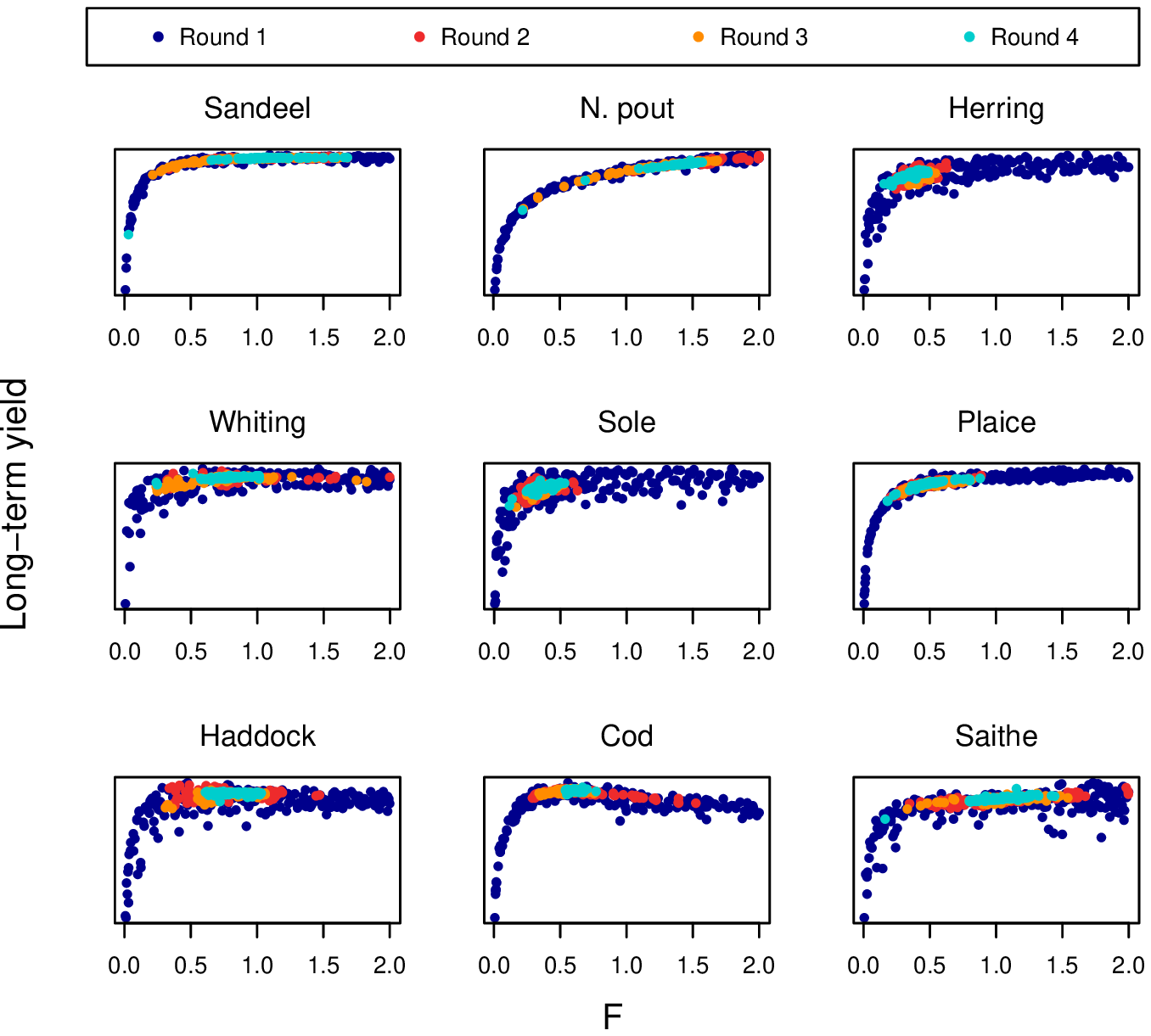}
\end{center}
\caption{The long-term yield predictions from mizer.}
\label{fig:miz_res}
\end{figure}

\begin{figure}[ht]
\begin{center}
    \includegraphics{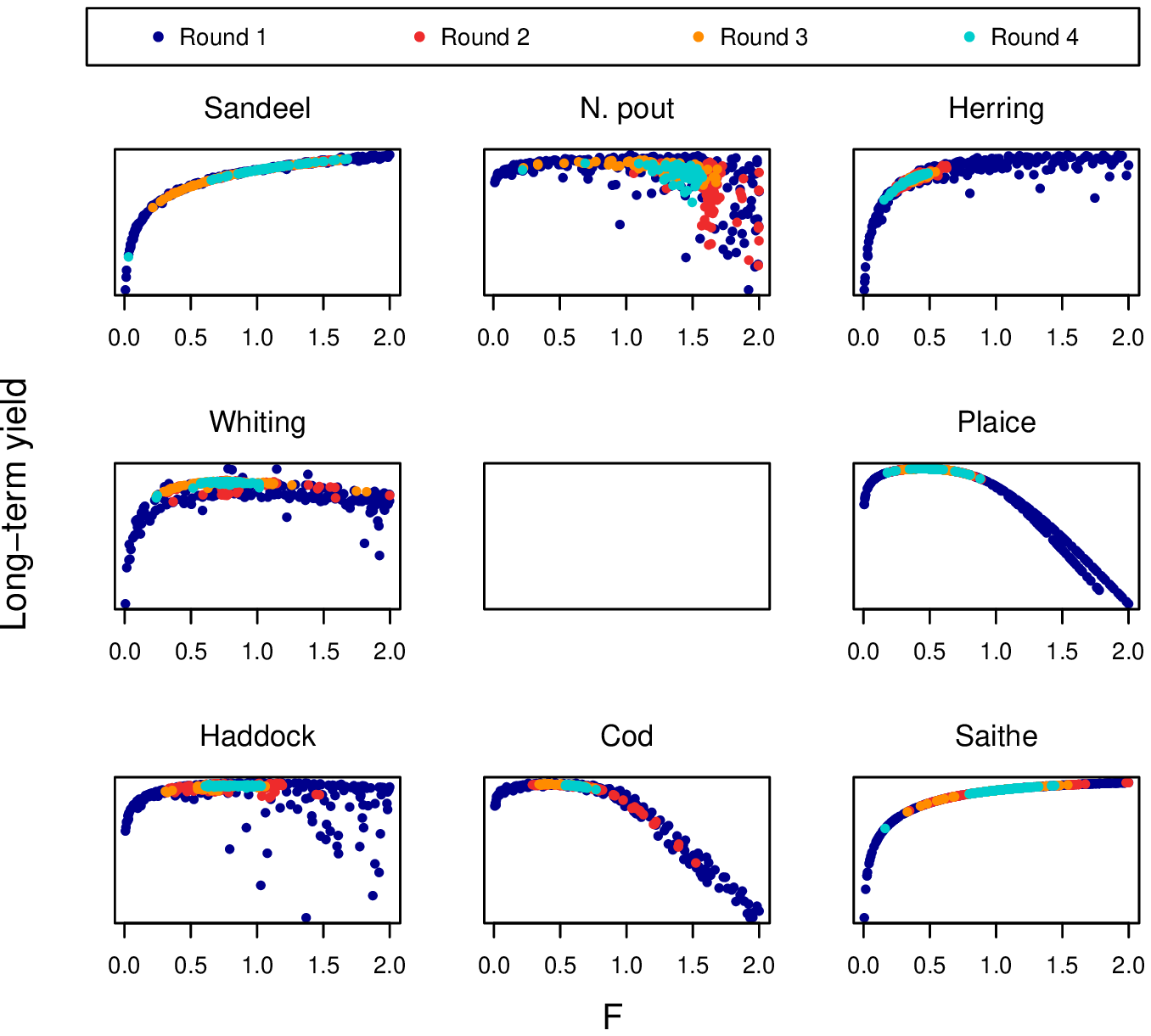}
\end{center}
\caption{The long-term yield predictions from FishSums.}
\label{fig:fs_res}
\end{figure}

\subsection{Spawning stock biomass}
Figure \ref{fig:ens_SSB} shows the 25th percentile of the long-term SSB. The solid line is the $B_{lim}$ for each species.

\subsection{Reference points}
Table \ref{tb:Nash} shows the 39 Nash equilibria and their expected long-term revenue that satisfy have acceptable risk to the species long-term SSB.

\begin{landscape}
\begin{longtable}{@{\extracolsep{2pt}} cccccccccc} 
  \caption{The 39 values of $\bm{F}_{Nash}$ and their expected long-term revenue that we found in this study.} 
  \label{tb:Nash} \\
\hline \\
Sandeel & N.pout & Herring & Whiting & Sole & Plaice & Haddock & Cod & Saithe & Revenue (\pounds{}billions) \\ 
\hline \\[-1.8ex] 
$1.05$ & $1.49$ & $0.46$ & $0.82$ & $0.31$ & $0.48$ & $0.94$ & $0.62$ & $1.10$ & $2.16$ \\ 
$1.10$ & $1.47$ & $0.44$ & $0.87$ & $0.31$ & $0.48$ & $0.76$ & $0.63$ & $1.16$ & $2.15$ \\ 
$1.11$ & $1.44$ & $0.39$ & $0.85$ & $0.37$ & $0.51$ & $0.82$ & $0.63$ & $1.13$ & $2.12$ \\ 
$1.04$ & $1.42$ & $0.40$ & $0.78$ & $0.31$ & $0.49$ & $0.97$ & $0.64$ & $1.09$ & $2.11$ \\ 
$1.39$ & $1.41$ & $0.38$ & $0.86$ & $0.31$ & $0.44$ & $0.86$ & $0.66$ & $0.97$ & $2.10$ \\ 
$1.06$ & $1.38$ & $0.41$ & $0.77$ & $0.27$ & $0.50$ & $0.80$ & $0.64$ & $0.83$ & $2.09$ \\ 
$0.93$ & $1.40$ & $0.46$ & $0.82$ & $0.35$ & $0.51$ & $0.77$ & $0.62$ & $1.16$ & $2.09$ \\ 
$1.31$ & $1.44$ & $0.44$ & $0.82$ & $0.37$ & $0.39$ & $0.87$ & $0.65$ & $0.98$ & $2.09$ \\ 
$1.01$ & $1.38$ & $0.48$ & $0.78$ & $0.30$ & $0.46$ & $0.87$ & $0.60$ & $0.98$ & $2.08$ \\ 
$1.12$ & $1.39$ & $0.41$ & $0.81$ & $0.33$ & $0.53$ & $0.72$ & $0.65$ & $0.98$ & $2.08$ \\ 
$1.10$ & $1.53$ & $0.47$ & $0.74$ & $0.37$ & $0.50$ & $0.94$ & $0.61$ & $1.11$ & $2.08$ \\ 
$1.08$ & $1.44$ & $0.44$ & $0.77$ & $0.35$ & $0.48$ & $0.79$ & $0.62$ & $0.96$ & $2.07$ \\ 
$1.04$ & $1.16$ & $0.36$ & $0.79$ & $0.36$ & $0.55$ & $0.90$ & $0.65$ & $1.24$ & $2.05$ \\ 
$0.91$ & $1.39$ & $0.42$ & $0.74$ & $0.27$ & $0.40$ & $0.95$ & $0.58$ & $0.93$ & $2.05$ \\ 
$1.10$ & $1.39$ & $0.47$ & $0.76$ & $0.31$ & $0.51$ & $0.69$ & $0.60$ & $0.93$ & $2.05$ \\ 
$1.20$ & $1.42$ & $0.43$ & $0.76$ & $0.36$ & $0.51$ & $0.73$ & $0.63$ & $0.94$ & $2.05$ \\ 
$1.14$ & $1.45$ & $0.46$ & $0.79$ & $0.43$ & $0.42$ & $0.84$ & $0.63$ & $1.00$ & $2.04$ \\ 
$1.02$ & $1.36$ & $0.41$ & $0.76$ & $0.32$ & $0.18$ & $1.00$ & $0.64$ & $1.12$ & $2.03$ \\ 
$0.98$ & $1.36$ & $0.38$ & $0.79$ & $0.32$ & $0.23$ & $0.83$ & $0.66$ & $1.07$ & $2.03$ \\ 
$1.01$ & $1.43$ & $0.45$ & $0.74$ & $0.39$ & $0.48$ & $0.74$ & $0.59$ & $0.91$ & $2.03$ \\ 
$1.33$ & $1.38$ & $0.40$ & $0.78$ & $0.31$ & $0.48$ & $0.69$ & $0.59$ & $1.02$ & $2.02$ \\ 
$1.16$ & $1.44$ & $0.42$ & $0.82$ & $0.14$ & $0.49$ & $0.66$ & $0.63$ & $1.15$ & $2.02$ \\ 
$1.01$ & $1.41$ & $0.42$ & $0.70$ & $0.36$ & $0.56$ & $0.76$ & $0.56$ & $1.03$ & $2.01$ \\ 
$1.04$ & $1.39$ & $0.42$ & $0.75$ & $0.35$ & $0.44$ & $0.63$ & $0.62$ & $0.99$ & $2.01$ \\ 
$1.05$ & $1.24$ & $0.40$ & $0.74$ & $0.33$ & $0.43$ & $0.70$ & $0.65$ & $1.24$ & $2.01$ \\ 
$0.90$ & $1.35$ & $0.38$ & $0.71$ & $0.34$ & $0.44$ & $0.79$ & $0.65$ & $0.83$ & $2.01$ \\ 
$0.90$ & $1.30$ & $0.46$ & $0.72$ & $0.36$ & $0.49$ & $0.72$ & $0.57$ & $0.96$ & $1.98$ \\ 
$0.94$ & $1.46$ & $0.40$ & $0.65$ & $0.27$ & $0.48$ & $0.83$ & $0.63$ & $1.14$ & $1.98$ \\ 
$1.20$ & $1.37$ & $0.38$ & $0.69$ & $0.41$ & $0.59$ & $0.79$ & $0.56$ & $0.94$ & $1.98$ \\ 
$1.08$ & $1.20$ & $0.43$ & $0.70$ & $0.41$ & $0.45$ & $0.80$ & $0.63$ & $0.83$ & $1.97$ \\ 
$1.07$ & $1.43$ & $0.39$ & $0.65$ & $0.33$ & $0.43$ & $0.63$ & $0.55$ & $1.20$ & $1.95$ \\ 
$0.78$ & $1.35$ & $0.39$ & $0.71$ & $0.43$ & $0.42$ & $1.02$ & $0.66$ & $1.19$ & $1.92$ \\ 
$1.44$ & $1.47$ & $0.35$ & $0.73$ & $0.32$ & $0.46$ & $0.92$ & $0.64$ & $1.26$ & $1.91$ \\ 
$1.49$ & $1.35$ & $0.38$ & $0.78$ & $0.33$ & $0.55$ & $0.69$ & $0.63$ & $1.23$ & $1.90$ \\ 
$1.66$ & $1.43$ & $0.41$ & $0.80$ & $0.35$ & $0.69$ & $0.77$ & $0.61$ & $1.00$ & $1.88$ \\ 
$1.68$ & $1.42$ & $0.38$ & $0.82$ & $0.33$ & $0.52$ & $0.90$ & $0.65$ & $1.18$ & $1.83$ \\ 
$0.71$ & $1.40$ & $0.49$ & $0.58$ & $0.38$ & $0.47$ & $0.64$ & $0.55$ & $0.87$ & $1.81$ \\ 
$1.38$ & $0.69$ & $0.30$ & $0.72$ & $0.42$ & $0.57$ & $0.96$ & $0.67$ & $1.44$ & $1.75$ \\ 
$0.66$ & $1.28$ & $0.23$ & $0.69$ & $0.36$ & $0.51$ & $0.67$ & $0.60$ & $1.25$ & $1.75$ \\ 
\hline \\[-1.8ex] 
\end{longtable}
\end{landscape}

\begin{figure}[ht]
\begin{center}
    \includegraphics{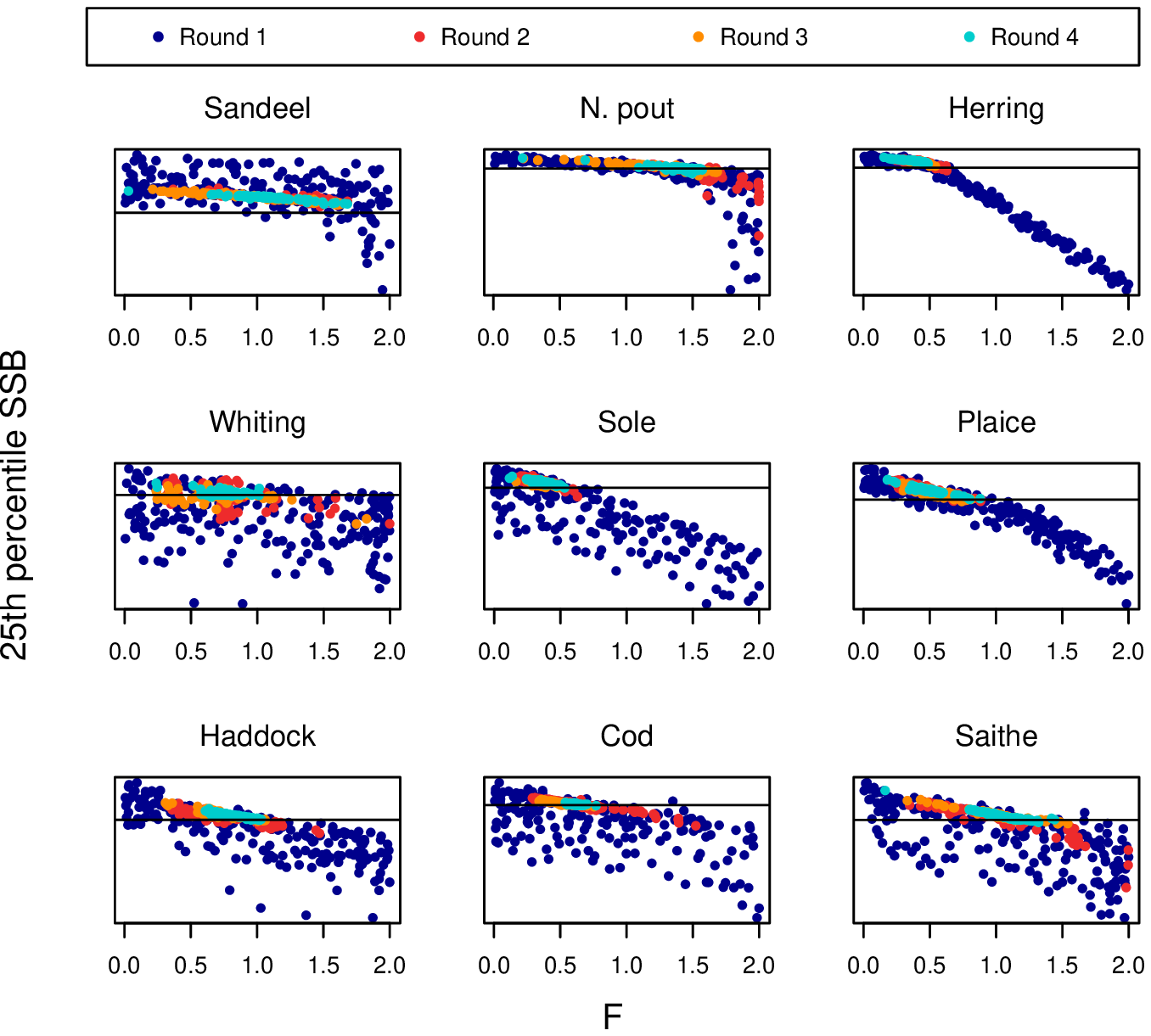}
\end{center}
\caption{The 25th percentile of the long-term spawning stock biomass. The solid line is the value for $B_{lim}$.}
\label{fig:ens_SSB}
\end{figure}

\subsection{Value of the yield}
Figure \ref{fig:revenue} shows the value of the yield for the 40 final Nash equilibria.

\begin{figure}[ht]
\begin{center}
    \includegraphics{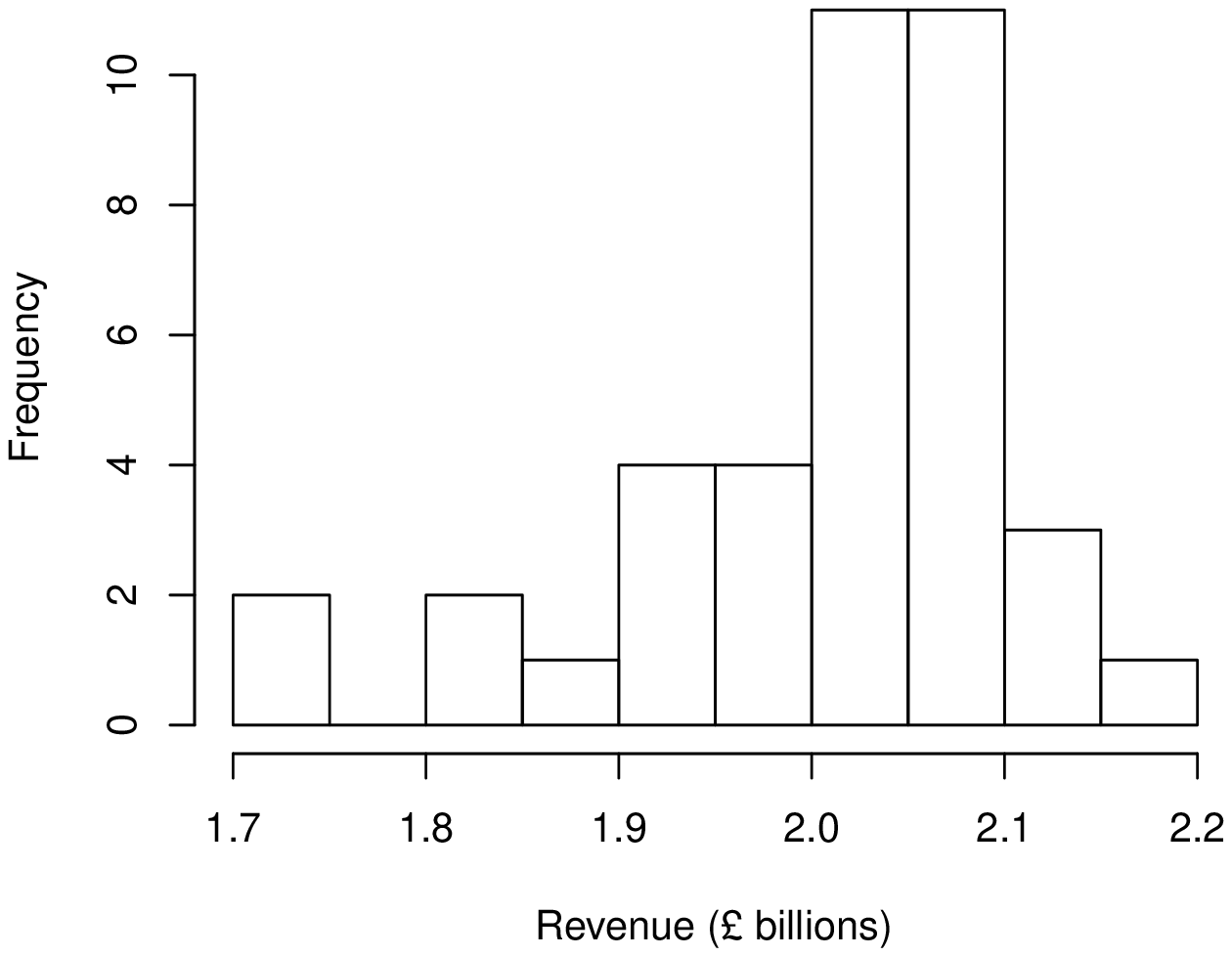}
\end{center}
\caption{The future annual revenue for the final Nash equilibria.}
\label{fig:revenue}
\end{figure}

\end{document}